\documentclass[11pt]{article}
\usepackage{liyang}

\usepackage[mathscr]{euscript}

\makeatletter
\newtheorem*{rep@theorem}{\rep@title}
\newcommand{\newreptheorem}[2]{
\newenvironment{rep#1}[1]{
 \def\rep@title{#2 \ref{##1}}
 \begin{rep@theorem}\itshape}
 {\end{rep@theorem}}}
\makeatother

\newreptheorem{theorem}{Theorem}
\newreptheorem{lemma}{Lemma}
\newreptheorem{proposition}{Proposition}

\begin{document}

\title{A composition theorem for the Fourier Entropy-Influence
  conjecture}

\author{Ryan O'Donnell \\ 
Carnegie Mellon University \\ 
{\tt odonnell@cs.cmu.edu}
\thanks{Supported by NSF grants
    CCF-0747250 and CCF-1116594, and a Sloan fellowship. This material
    is based upon work supported by the National Science Foundation
    under grant numbers listed above. Any opinions, findings and
    conclusions or recommendations expressed in this material are
    those of the author and do not necessarily reflect the views of
    the National Science Foundation (NSF).}  \and Li-Yang Tan \\
    Columbia University \\ 
    {\tt liyang@cs.columbia.edu}
  \thanks{Research done while visiting CMU.}}  


\maketitle

\begin{abstract}
The Fourier Entropy-Influence (FEI) conjecture of Friedgut and
Kalai \cite{FK96b} seeks to relate two fundamental measures of Boolean
function complexity: it states that $\bH[f] \leq C\cdot \Inf[f]$ holds
for every Boolean function $f$, where $\bH[f]$ denotes the spectral
entropy of $f$, $\Inf[f]$ is its total influence, and $C > 0$ is a
universal constant.  Despite significant interest in the conjecture it
has only been shown to hold for a few classes of Boolean
functions. \medskip

Our main result is a composition theorem for the FEI conjecture.  We
show that if $g_1,\ldots,g_k$ are functions over disjoint sets of
variables satisfying the conjecture, and if the Fourier transform of
$F$ taken with respect to the product distribution with biases
$\E[g_1], \ldots, \E[g_k]$ satisfies the conjecture, then their
composition $F(g_1(x^1),\ldots,g_k(x^k))$ satisfies the conjecture. As
an application we show that the FEI conjecture holds for read-once
formulas over arbitrary gates of bounded arity, extending a recent
result~\cite{OWZ11b} which proved it for read-once decision trees.
Our techniques also yield an explicit function with the largest known
ratio of $C \geq 6.278$ between $\bH[f]$ and $\Inf[f]$, improving on
the previous lower bound of $4.615$.
\end{abstract}

\newcommand{\lnote}[1]{\footnote{\color{blue}Li-Yang : {#1}}}
\newcommand{\todo}[1]{{\color{red} #1}}

\section{Introduction}

A longstanding and important open problem in the field of Analysis of
Boolean Functions is the Fourier Entropy-Influence conjecture
made by Ehud Friedgut and Gil Kalai in 1996~\cite{FK96b,Kal07}. The
conjecture seeks to relate two fundamental analytic measures of
Boolean function complexity, the spectral entropy and total
influence:\medskip

\noindent {\bf Fourier Entropy-Influence (FEI) Conjecture.} {\it
There exists a universal constant $C > 0$ such that for every Boolean
function $f\isafunc$, it holds that $\bH[f]\leq C\cdot \Inf[f]$. That
is,
\[ \sumS \hatf(S)^2 \log_2\left(\frac1{\hatf(S)^2}\right) \leq C \sumS
|S|\cdot \hatf(S)^2. \] }

Applying Parseval's identity to a Boolean function $f$ we get
$\sumS\hatf(S)^2 = \E[f(\bx)^2] = 1$, and so the Fourier coefficients
of $f$ induce a probability distribution $ \mathscr{S}_f$ over the
$2^n$ subsets of $[n]$ wherein $S\sse [n]$ has ``weight'' (probability mass)
$\hatf(S)^2$. The \emph{spectral entropy} of $f$, denoted $\bH[f]$, is
the Shannon entropy of $\mathscr{S}_f$, quantifying how spread out the
Fourier weight of $f$ is across all $2^n$ monomials. The influence of
a coordinate $i\in [n]$ on $f$ is $\Inf_i[f] =
\Pr[f(\bx)\neq f(\bx^{\oplus i})]$\footnote{All probabilities and
  expectations are with respect to the uniform distribution unless
  otherwise stated.}, where $\bx^{\oplus i}$ denotes $\bx$ with its $i$-th
bit flipped, and the \emph{total influence} of $f$ is simply $\Inf[f]
= \sumi\Inf_i[f]$.  Straightforward Fourier-analytic calculations show
that this combinatorial definition is equivalent to the
quantity $\E_{\bS\sim\mathscr{S}_f}[|\bS|] = \sumS |S|\cdot\hatf(S)^2$,
and so total influence measures the degree distribution of the
monomials of $f$, weighted by the squared-magnitude of its coefficients.
Roughly speaking then, the FEI conjecture states that a Boolean
function whose Fourier weight is well ``spread out'' ({\it i.e.}
has high spectral entropy) must have a significant portion of its
Fourier weight lying on high degree monomials ({\it i.e.} have
high total influence).\footnote{The assumption that $f$
  is Boolean-valued is crucial here, as the same conjecture is false
  for functions $f:\bn\to\R$ satisfying $\sumS\hatf(S)^2 = 1$.  The
  canonical counterexample is $f(x) = \frac1{\sqrt{n}} \sumi x_i$
  which has total influence $1$ and spectral entropy $\log_2n$.}\medskip

In addition to being a natural question concerning the Fourier
spectrum of Boolean functions, the FEI conjecture also has important
connections to several areas of theoretical computer science and
mathematics. 
Friedgut and Kalai's original motivation was to understand general
conditions under which monotone graph properties exhibit sharp
thresholds, and the FEI conjecture captures the intuition that having
significant symmetry, hence high spectral entropy, is one such
condition.  Besides its applications in the study of random graphs,
the FEI conjecture is known to imply the celebrated Kahn-Kalai-Linial theorem
\cite{KKL88}
: \medskip

\noindent {\bf KKL Theorem.} \medskip
{\it For every Boolean function $f$ there exists an $i\in [n]$ such that
$\Inf_i[f] = \Var[f]\cdot \Omega(\frac{\log n}{n})$.} \medskip

The FEI conjecture also implies Mansour's conjecture \cite{Man94}: \medskip

\noindent {\bf Mansour's Conjecture.} {\it Let $f$ be a Boolean
  function computed by a $t$-term DNF formula. For any constant
  $\eps > 0$ there exists a
  collection $\calS\sse 2^{[n]}$ of cardinality $\poly(t)$ such that
  $\sum_{S\in\calS}\hatf(S)^2 \geq 1-\eps$.}\medskip

    Combined with recent work of Gopalan
    {\it et al.}\ \cite{GKK08}, Mansour's conjecture yields an
    efficient algorithm for agnostically learning the class of
    $\poly(n)$-term DNF formulas from queries. This would resolve a
    central open problem in computational learning theory
    \cite{GKK08a}. De \etal also noted that sufficiently strong
    versions of Mansour's conjecture would yield improved pseudorandom
    generators for depth-$2$ ${\sf AC}^0$ circuits \cite{DETT10}.
    More generally, the FEI conjecture implies the existence of sparse
    $L_2$-approximators for Boolean functions with small total
    influence: \medskip

    \noindent {\bf Sparse $L_2$-approximators.} Assume the FEI
    conjecture holds. Then for every Boolean function $f$ there exists
    a $2^{O(\Inf[f]/\eps)}$-sparse polynomial $p:\R^n\to\R$ such that
    $\E[(f(\bx)-p(\bx))^2]\leq \eps$.\medskip

    By Friedgut's junta theorem \cite{Fri98}, the above holds
    unconditionally with a weaker bound of
    $2^{O(\Inf[f]^2/\eps^2)}$. This is the main technical ingredient
    underlying several of the best known uniform-distribution
    learning algorithms \cite{Ser04a,OS08b}.\medskip

For more on the FEI conjecture we refer the reader to
Kalai's blog post~\cite{Kal07}.

\subsection{Our results}

Our research is motivated by the following question:

\begin{question}
\label{q:compose}
 Let $F:\bits^k\to\bits$ and
  $g_1,\ldots,g_k:\bits^\ell\to\bits$.  What properties do $F$ and
  $g_1,\ldots,g_k$ have to satisfy for the FEI conjecture to
  hold for the disjoint composition $f(x^1,\ldots,x^k) =
  F(g_1(x^1),\ldots,g_k(x^k))$?
\end{question}

Despite its simplicity this question has not been well understood.
For example, prior to our work the FEI conjecture was open even for
read-once DNFs (such as the ``tribes'' function); these are
the disjoint compositions of $F={\sf OR}$ and $g_1,\ldots,g_k={\sf
  AND}$, perhaps two of the most basic Boolean functions with
extremely simple Fourier spectra.  Indeed, Mansour's conjecture, a
weaker conjecture than FEI, was only recently shown to hold for
read-once DNFs~\cite{KLW10,DETT10}.  Besides being a fundamental
question concerning the behavior of spectral entropy and total
influence under composition, Question \ref{q:compose} (and our answer
to it) also has implications for a natural approach towards disproving
the FEI conjecture; we elaborate on this at the end of this
section. \medskip

  A particularly appealing and general
answer to Question~\ref{q:compose} that one may hope for would be the following: ``if $\bH[F]\leq C_1\cdot \Inf[F]$ and $\bH[g_i] \leq C_2 \cdot
\Inf[g_i]$ for all $i\in [k]$, then $\bH[f] \leq \max\{C_1,C_2\}\cdot
\Inf[f]$.''
While this is easily seen to be false\footnote{For example, by
  considering $F = {\sf
  OR}_2$, the $2$-bit disjunction, and $g_1,g_2 = {\sf AND}_2$, the
$2$-bit conjunction.},  
our main result shows that this proposed answer to Question
\ref{q:compose} \emph{is} in fact true for a carefully chosen
sharpening of the FEI conjecture. To arrive at a formulation that
bootstraps itself, we first consider a slight strengthening of the FEI
conjecture which we call FEI$^+$, and then work with a generalization
of FEI$^+$ that concerns the Fourier spectrum of $f$ not just with
respect to the uniform distribution, but an arbitrary product
distribution over $\bn$:

\begin{conjecture}[FEI$^+$ for product distributions]
\label{conj:prod-FEI++}
There is a universal constant $C > 0$ such that the following holds.
Let $\mu=\la \mu_1,\ldots,\mu_n\ra$ be any sequence of biases and
$f:\bn_\mu\to\bits$.  Here the notation $\bn_\mu$ means that we think
of $\bn$ as being endowed with the $\mu$-biased product probability
distribution in which $\Ex_\mu[x_i] = \mu_i$ for all $i\in [n]$.  Let
$\{\wtf(S)\}_{S\sse [n]}$ be the $\mu$-biased Fourier coefficients of
$f$. Then
\[ \sum_{S\neq\emptyset}\wtf(S)^2\log\left(\frac{\prod_{i\in
      S}(1-\mu_i^2)}{\wtf(S)^2}\right) \leq C \cdot
  (\Inf^\mu[f]-\Varx_\mu[f]).  \]
\end{conjecture}
\noindent We write $\bH^\mu[f]$ to denote the quantity $\sumS
\wtf(S)^2\log\left(\prod_{i\in
      S}(1-\mu_i^2)/\wtf(S)^2\right)$, and so the inequality of
Conjecture \ref{conj:prod-FEI++} can be equivalently stated as
$\bH^\mu[f^{\geq 1}] \leq C\cdot (\Inf^\mu[f]-\Var_\mu[f])$.\medskip

In Proposition \ref{prop:fei+-implies-fei} we show that Conjecture
\ref{conj:prod-FEI++} with $\mu = \la 0,\ldots,0\ra$ (the uniform distribution) implies the FEI conjecture.  We say that a
Boolean function $f$ ``satisfies $\mu$-biased FEI$^+$ with factor
  $C$'' if the $\mu$-biased Fourier transform of $f$ satisfies the
inequality of Conjecture \ref{conj:prod-FEI++}.  Our main result,
which we prove in Section \ref{sec:composition-lemma}, is a
composition theorem for FEI$^+$:

\begin{theorem}
\label{thm:compose}
  Let $f(x^1,\ldots,x^k) = F(g_1(x^1),\ldots,g_k(x^k))$, where the
  domain of $f$ is endowed with a product distribution $\mu$. Suppose
  $g_1,\ldots,g_k$ satisfy $\mu$-biased FEI$^+$ with factor $C_1$ and
  $F$ satisfies $\eta$-biased FEI$^+$ with factor $C_2$, where $\eta =
  \la \Ex_\mu[g_1],\ldots,\Ex_\mu[g_k]\ra$.  Then $f$ satisfies
  $\mu$-biased FEI$^+$ with factor $\max\{C_1,C_2\}$.
\end{theorem}


Theorem \ref{thm:compose} suggests an inductive approach towards
proving the FEI conjecture for read-once de Morgan formulas: since the
dictators $\pm x_i$ trivially satisfy uniform-distribution FEI$^+$
with factor $1$, it suffices to prove that both ${\sf AND}_2$ and
${\sf OR}_2$ satisfy $\mu$-biased FEI$^+$ with some constant
\emph{independent of} $\mu\in [-1,1]^2$. In Section
\ref{sec:product-FEI-bound} we prove that in fact \emph{every}
$F:\bits^k\to\bits$ satisfies $\mu$-biased FEI$^+$ with a factor depending
only on its arity $k$ and not the biases $\mu_1,\ldots,\mu_k$.

\begin{theorem}
\label{thm:product-FEI-bound}
  Every $F:\bits^k\to \bits$ satisfies $\mu$-biased FEI$^+$ with factor $C =
  2^{O(k)}$ for any product distribution $\mu = \la
  \mu_1,\ldots,\mu_k\ra$.
\end{theorem}

Together, Theorems \ref{thm:compose} and \ref{thm:product-FEI-bound}
imply:

\begin{theorem}
\label{thm:read-once}
  Let $f$ be computed by a read-once formula over the basis~$\calB$
  and~$\mu$ be any sequences of biases. Then $f$ satisfies
  $\mu$-biased
  FEI$^+$ with factor $C$, where~$C$ depends only on the arity of the
  gates in~$\calB$.
\end{theorem}

Since uniform-distribution FEI$^+$ is a strengthening of the FEI
conjecture, Theorem \ref{thm:read-once} implies that the FEI
conjecture holds for read-once formulas over arbitrary gates of
bounded arity.  As mentioned above, prior to our work the FEI
conjecture was open even for the class of read-once DNFs, a small
subclass of read-once formulas over the de Morgan basis $\{{\sf
  AND}_2,{\sf OR}_2,{\sf NOT}\}$ of arity $2$.  Read-once formulas
over a rich basis $\calB$ are a natural generalization of read-once de
Morgan formulas, and have seen previous study in concrete complexity (see
\eg \cite{HNW93}). \bigskip

\noindent {\bf Improved lower bound on the FEI constant.}
Iterated disjoint composition is commonly used to achieve separations
between complexity measures for Boolean functions \cite{BdW02}, and
 represents a natural approach towards disproving the FEI
conjecture. For example, one may seek a function $F$ such that
iterated compositions of $F$ with itself achieves a super-constant
amplification of the ratio between $\bH[F]$ and $\Inf[F]$, or consider
variants such as iterating $F$ with a different combining function
$G$.  Theorem \ref{thm:read-once} rules out as potential
counterexamples all such constructions based on iterated
composition. \medskip

However, the tools we develop to prove Theorem~\ref{thm:read-once} also yield an explicit function $f$ achieving the
best-known separation between $\bH[f]$ and $\Inf[f]$ (\ie the constant $C$
in the statement of the FEI conjecture).  In Section
\ref{sec:lb-unif-FEI} we prove:

\begin{theorem}
\label{thm:lb-unif-FEI}
There exists an explicit family of functions $f_n\isafunc$ such that
\[ \lim_{n\to\infty}\frac{\bH[f_n]}{\Inf[f_n]} \geq 6.278. \]
\end{theorem}

\noindent This improves on the previous lower bound of $C \ge 60/13 \approx
4.615$ \cite{OWZ11b}. \medskip

\noindent {\bf Previous work.} The first published progress on the
FEI conjecture was by Klivans {\it et al.} who proved the
conjecture for random $\poly(n)$-term DNF formulas
\cite{KLW10}. This was followed by the work of O'Donnell {\it et al.}
who proved the conjecture for the class of symmetric functions and
read-once decision trees \cite{OWZ11b}.\medskip

The FEI conjecture for product distributions was studied in the recent
work of Keller \etal \cite{KMS12}, where they consider the case of all
the biases being the same.  They introduce the following
generalization of the FEI conjecture to these measures, and show via a
reduction to the uniform distribution \cite{BKK+92} that it is
equivalent to the FEI conjecture:

\begin{conjecture}[Keller-Mossel-Schlank] There is a universal constant $C$ such that the following holds. Let $0 < p < 1$ and
  $f:\bn\to\bits$, where the domain of $f$ is endowed with the product
  distribution where $\Pr[x_i = -1] = p$ for all $i\in
  [n]$. Let $\{\wtf(S)\}_{S\sse[n]}$ be the Fourier coefficients of
    $f$ with respect to this distribution. Then
\[ \sumS \wtf(S)^2\log_2\left(\frac1{\wtf(S)^2}\right) \leq
C \cdot \frac{\log(1/p)}{1-p}\sumS |S| \cdot \wtf(S)^2. \]
\end{conjecture}

Notice that in this conjecture, the constant on the right-hand side, $C \cdot \frac{\log(1/p)}{1-p}$, depends on~$p$. By way of contrast, in our Conjecture~\ref{conj:prod-FEI++} the right-hand side constant has no dependence on~$p$; instead, the dependence on the biases is built into the definition of spectral entropy.  We view our generalization of the FEI conjecture to
arbitrary product distributions (where the biases are not necessarily
identical) as a key contribution of this work, and point to
our composition theorem as evidence in favor of Conjecture
\ref{conj:prod-FEI++} being a good statement to work with.
\section{Preliminaries}
 \label{sec:biased-fourier}
\noindent{\bf Notation.}  We will be concerned with functions
$f:\bn_\mu\to\R$ where $\mu = \la \mu_1,\ldots,\mu_n\ra \in [0,1]^n$
is a sequence of biases. Here the notation $\bn_\mu$ means that we
think of $\bn$ as being endowed with the $\mu$-biased product
probability distribution in which $\Ex_\mu[x_i] = \mu_i$ for all $i\in
[n]$.  We write $\sigma_i^2$ to denote variance of the $i$-th
coordinate $\Var_\mu[x_i] = 1-\mu_i^2$, and $\varphi : \R\to\R$ as
shorthand for the function $t\mapsto t^2\log(1/t^2)$, adopting the
convention that $\varphi(0) = 0$. We will assume familiarity with the
basics of Fourier analysis with respect to product distributions over
$\bn$; a review is included in Appendix \ref{ap:biased-fourier}.
 
\begin{proposition}[FEI$^+$ implies FEI]
\label{prop:fei+-implies-fei}
Suppose $f$ satisfies uniform-distribution FEI$^+$ with factor
$C$. Then $f$ satisfies the FEI conjecture with factor
$\max\{C,1/\ln 2\}$. \end{proposition}

\begin{proof}
  Let $\hatf(\emptyset)^2 = 1-\eps$, where $\eps = \Var[f]$ by
  Parseval's identity. By our assumption that $f$ satisfies
  uniform-distribution FEI$^+$ with factor $C$, we have
\begin{eqnarray*}
  \sumS\hatf(S)^2 \log\left(\frac{\prod_{i\in
        S}\sigma_i^2}{\hatf(S)^2}\right) & \leq & C\cdot (
  \Inf[f] -\Var[f]) + (1-\eps)\log
  \frac1{(1-\eps)} \\
  & \leq & C\cdot ( \Inf[f] -\Var[f]) + \frac{\eps}{\ln 2} \\
  &=  & C \cdot \Inf[f] + \left(\frac{1}{\ln 2}-C\right)\cdot \Var[f].
\end{eqnarray*}
If $C > 1/\ln 2$ then the RHS is at most $C\cdot \Inf[f]$ since
$(\frac1{\ln 2}-C)\cdot \Var[f]$ is negative. Otherwise we apply the
Poincar\'e inequality (Theorem \ref{thm:poincare}) to conclude that
the RHS is at most $C\cdot \Inf[f] + (\frac1{\ln 2}-C)\cdot
\Inf[f] = \frac1{\ln 2}\cdot \Inf[f]$.
\end{proof}

\section{Composition theorem for FEI$^+$}
\label{sec:composition-lemma}
We will be concerned with compositions of functions $f =
F(g_1(x^1),\ldots,g_k(x^k))$ where $g_1,\ldots,g_k$ are over disjoint
sets of variables each of size $\ell$. The domain of each $g_i$ is
endowed with a product distribution $\mu^i = \la
\mu^i_1,\ldots,\mu^i_\ell \ra$, which induces an overall product
distribution $\mu = \la \mu^1_1,\ldots,\mu^1_\ell , \ldots,
\mu^k_1,\ldots, \mu^k_\ell\ra$ over the domain of
$f:\bits^{k\ell}\to\bits$.  For notational clarity we will adopt the
equivalent view of $g_1,\ldots,g_k$ as functions over the same domain
$\bits^{k\ell}_\mu$ endowed with the same product distribution $\mu$,
with each $g_i$ depending only on $\ell$ out of $k\ell$
variables. \medskip

Our first lemma gives formulas for the spectral entropy and total influence
of the product of functions
$\Phi_1,\ldots,\Phi_k$ over disjoint sets of variables. The lemma
holds for real-valued functions $\Phi_i$; we
require this level of generality as we will not be applying the lemma
directly to the Boolean-valued functions $g_1,\ldots,g_k$ in the
composition $F(g_1(x^1),\ldots,g_k(x^k))$, but instead to their
normalized variants $\Phi(g_i) = (g_i-\E[g_i])/\Var[g_i]^{1/2}$.

\begin{lemma}
\label{lem:ent-inf-tensor}
Let $\Phi_1,\ldots,\Phi_k:\bits^{k\ell}_\mu\to\R$ where each $\Phi_i$
depends only on the $\ell$ coordinates in
$\{(i-1)\ell+1,\ldots,i\ell\}$.  Then
\[
 \bH^\mu[\Phi_1\cdots \Phi_k] = \sum_{i=1}^k
  \bH^\mu[\Phi_i]\prod_{j\neq i} \Ex_\mu[\Phi_j^2]
 \ \text{and} \ \
 \Inf^\mu[\Phi_1\cdots \Phi_k] =
 \sum_{i=1}^k \Inf^\mu[\Phi_i]
  \prod_{j\neq i} \Ex_\mu[\Phi_j^2]. \]
\end{lemma}

Due to space considerations we defer the proof of Lemma~\ref{lem:ent-inf-tensor} to Appendix \ref{ap:omitted}. We note that
this lemma recovers as a special case the folklore observation that
the FEI conjecture ``tensorizes'': for any $f$ if we define $f^{\oplus
  k}(x^1,\ldots,x^k) = f(x^1)\cdots f(x^k)$ then $\bH[f^{\oplus k}] =
k\cdot \bH[f]$ and $\Inf[f^{\oplus k}] = k\cdot \Inf[f]$. Therefore
$\bH[f] \leq C\cdot \Inf[f]$ if and only if $\bH[f^{\oplus k}] \leq
C\cdot \Inf[f^{\oplus k}]$.\medskip

Our next proposition relates the basic analytic measures -- spectral
entropy, total influence, and variance -- of a composition $f =
F(g_1(x^1),\ldots,g_k(x^k))$ to the corresponding quantities of the
combining function $F$ and base functions $g_1,\ldots,g_k$.  As
alluded to above, we accomplish this by considering $f$ as a linear
combination of the normalized functions $\Phi(g_i) =
(g_i-\E[g_i])/\Var[g_i]^{1/2}$ and applying Lemma
\ref{lem:ent-inf-tensor} to each term in the sum.  We mention that
this proposition is also the crux of our new lower bound of $C\geq
6.278$ on the constant of the FEI conjecture,
which we present in Section \ref{sec:lb-unif-FEI}.

\begin{proposition}
\label{prop:ent-inf-var-compose}
Let $F:\bits^k\to\R$, and $g_1,\ldots,g_k:\bits^{k\ell}_\mu\to\bits$ where
each $g_i$ depends only on the $\ell$ coordinates in
$\{(i-1)\ell+1,\ldots,i\ell\}$. Let $f(x) = F(g_1(x),\ldots,g_k(x))$ and
$\{\wt{F}(S)\}_{S\sse [k]}$ be the $\eta$-biased Fourier coefficients of
$F$ where $\eta =
  \la \Ex_\mu[g_1]), \ldots,\Ex_\mu[g_k] \ra$. Then
\begin{eqnarray}
   \bH^\mu[f^{\geq 1}] &=& \bH^\eta[F^{\geq 1}] +\sum_{S\neq\emptyset}
  \wt{F}(S)^2 \sum_{i\in S}\frac{\bH^\mu[g_i^{\geq
      1}]}{\Varx_\mu[g_i]}, \label{eq:ent-compose} \\
\Inf^\mu[f] &=& \sum_{S\neq\emptyset} \wt{F}(S)^2 \sum_{i\in S}
\frac{\Inf^\mu[g_i]}{\Varx_\mu[g_i]}, \quad and  \label{eq:inf-compose} \\
\Var_\mu[f] &=& \sum_{S\neq\emptyset}\wt{F}(S)^2 =
\Var_\eta[F]. \label{eq:var-compose}
\end{eqnarray}
\end{proposition}

\begin{proof}
  By the $\eta$-biased Fourier expansion of $F:\bits^k_\eta\to\R$
 and the definition of $\eta$
  we have
  \[ F(y_1,\ldots,y_k) = \sumS \wt{F}(S) \prod_{i\in S} \frac{y_i -
    \eta_i}{\sqrt{1-\eta_i^2}} = \sumS \wt{F}(S) \prod_{i\in S}
  \frac{y_i-\Ex_\mu[g_i]}{\Var_\mu[g_i]^{1/2}},\] so we may write
\[ F(g_1(x),\ldots,g_k(x)) = \sumS\wt{F}(S)
\prod_{i\in S}\Phi(g_i(x)), \  \text{where}\ \Phi(g_i(x)) =
\frac{g_i(x)-\Ex_\mu[g_i]}{\Var_\mu[g_i]^{1/2}}. \]
Note that $\Phi$ normalizes $g_i$ such that $\Ex_\mu[\Phi(g_i)] = 0$ and
$\Ex_\mu[\Phi(g_i)^2] = 1$.  First we claim that
\[ \bH^\mu[f^{\geq 1}] =\bH^\mu\bigg[\sum_{S \neq\emptyset}
\wt{F}(S)\prod_{i\in S}\Phi(g_i)\bigg] = \sum_{S \neq\emptyset}
\bH^\mu\Big[\wt{F}(S)\prod_{i\in S}\Phi(g_i)\Big].\] It suffices to show
that for any two distinct non-empty sets $S, T\sse [k]$, no monomial
$\phi_U^\mu$ occurs in the $\mu$-biased spectral support of both
$\wt{F}(S) \prod_{i\in S}\Phi(g_i)$ and $\wt{F}(T) \prod_{i\in
  T}\Phi(g_i)$.  To see this recall that $\Phi(g_i)$ is balanced with
respect to $\mu$ (\ie $\Ex_\mu[\Phi(g_i)] =
\Ex_\mu[\Phi(g_i)\phi^\mu_\emptyset] = 0$), and so every monomial
$\phi_U^\mu$ in the support of $\wt{F}(S)\prod_{i\in S}\Phi(g_i)$ is of
the form $\prod_{i\in S}\phi^\mu_{U_i}$ where $U_i$ is a non-empty
subset of the relevant variables of $g_i$ (\ie
$\{(i-1)\ell+1,\ldots,i\ell\}$); likewise for monomials in the support
of $\wt{F}(T)\prod_{i\in T}\Phi(g_i)$.  In other words the non-empty
subsets of $[k]$ induce a partition of the $\mu$-biased Fourier support
of $f$, where $\phi^\mu_U$ is mapped to $\emptyset\neq S\sse [k]$ if and
only if $U$ contains a relevant variable of $g_i$ for every $i\in S$
and none of the relevant variables of $g_j$ for any $j\notin
S$. \medskip

With this identity in hand we have
\begin{eqnarray*} \bH^\mu[f^{\geq 1}] &=& \sum_{S\neq\emptyset}
  \bH^\mu\Big[\wt{F}(S)
  \prod_{i\in S}\Phi(g_i)\Big] \\
&=& \sum_{S\neq\emptyset} \varphi(\wt{F}(S)) +
  \wt{F}(S)^2 \sum_{i\in S}\bH^\mu[\Phi(g_i)]. \\
  &=&\sum_{S\neq\emptyset} \varphi(\wt{F}(S)) + \wt{F}(S)^2 \sum_{i\in
    S}\left(\frac{\bH^\mu[g_i-\Ex_\mu[g_i]]}{\Varx_\mu[g_i]} +
  \varphi\left(\frac1{\Varx_\mu[g_i]^{1/2}}\right)\Varx_\mu[g_i]\right)
\\
  &=& \bH^\eta[F^{\geq 1}] +\sum_{S\neq\emptyset}
  \wt{F}(S)^2 \sum_{i\in S}\frac{\bH^\mu[g_i^{\geq
      1}]}{\Varx_\mu[g_i]},
\end{eqnarray*}
where the second and third equalities are two applications of Lemma
\ref{lem:ent-inf-tensor} (for the second equality we view $\wt{F}(S)$
as a constant function with $\bH^\mu[\wt{F}(S)] = \varphi(\wt{F}(S))$).
By the same reasoning, we also have
\begin{eqnarray*}
 \Inf^\mu[f] \ =\ \sum_{S\neq\emptyset} \Inf^\mu\Big[\wt{F}(S)\prod_{i\in
  S}\Phi(g_i(x^i))\Big]  & =& \sum_{S\neq\emptyset}  \wt{F}(S)^2
 \sum_{i\in S}\Inf^\mu[\Phi(g_i)]
 \\
&=& \sum_{S\neq\emptyset} \wt{F}(S)^2 \sum_{i\in S}
\frac{\Inf^\mu[g_i]}{\Varx_\mu[g_i]}.
\end{eqnarray*}
Here the second equality is by Lemma \ref{lem:ent-inf-tensor}, again
viewing $\wt{F}(S)$ as a constant function with $\Inf^\mu[\wt{F}(S)] =
0$, and the third equality uses the fact that $\Inf^\mu[\alpha f] =
\alpha^2\cdot  \Inf^\mu[f]$ and $\Inf^\mu[g_i-\Ex_\mu[g_i]] =
\Inf^\mu[g_i]$.  Finally we see that
\[ \Var_\mu[f] = \sum_{S\neq\emptyset}\Var_\mu\Big[\wt{F}(S)\prod_{i\in
  S}\Phi(g_i)\Big] = \sum_{S\neq\emptyset}\wt{F}(S)^2 \prod_{i\in
  S}\Var_\mu[\Phi(g_i)] = \sum_{S\neq\emptyset}\wt{F}(S)^2,\] where the
last quantity is $\Var_\eta[F]$.  Here the second equality uses the fact
that the functions $\Phi(g_i)$ are on disjoint sets of variables (and
therefore statistically independent when viewed as random variables),
and the third equality holds since $\Var_\mu[\Phi(g_i)] =
\E[\Phi(g_i)^2] - \E[\Phi(g_i)]^2 = 1$.
\end{proof}

We are now ready to prove our main theorem:

\begin{reptheorem}{thm:compose}
Let $F:\bits^k\to\R$, and $g_1,\ldots,g_k:\bits^{k\ell}_\mu\to\bits$ where
each $g_i$ depends only on the $\ell$ coordinates in
$\{(i-1)\ell+1,\ldots,i\ell\}$. Let $f(x) = F(g_1(x),\ldots,g_k(x))$ and
suppose $C > 0$ satisfies
\begin{enumerate}
\item $\bH^\mu[g_i^{\geq 1}] \leq C\cdot  (\Inf^\mu[g_i] - \Var_\mu[g_i])$
  for all
  $i\in [k]$.
\item $\bH^\eta[F^{\geq 1}] \leq C\cdot (\Inf^\eta[F] - \Var_\eta[F])$, where
  $\eta =\la \Ex_\mu[g_1], \ldots, \Ex_\mu[g_k]\ra$.
\end{enumerate}
Then $\bH^\mu[f^{\geq 1}] \leq C\cdot (\Inf^\mu[f]-\Var_\mu[f])$.
\end{reptheorem}

\begin{proof}
By our first assumption each $g_i$ satisfies $\Inf^\mu[g_i]\geq
\frac1{C}\bH^\mu[g^{\geq 1}] + \Varx_\mu[g_i]$, and so combining this with
equation (\ref{eq:inf-compose}) of Proposition
\ref{prop:ent-inf-var-compose} we have
\begin{eqnarray}
\Inf^\mu[f]\ = \  \sum_{S\neq\emptyset} \wt{F}(S)^2 \sum_{i\in S}
\frac{\Inf^\mu[g_i]}{\Varx_\mu[g_i]}
&\geq& \sum_{S\neq\emptyset} \wt{F}(S)^2 \sum_{i\in
  S}\left(\frac{\bH^\mu[g_i^{\geq 1}]}{C\Varx_\mu[g_i]} + 1 \right)
\nonumber \\
&=&   \Inf^\eta[F] + \frac1{C}\sum_{S\neq\emptyset}
\wt{F}(S)^2 \sum_{i\in S}\frac{\bH^\mu[g_i^{\geq
    1}]}{\Varx_\mu[g_i]} \quad\quad \label{eq:inf-induct}
\end{eqnarray}
This along with equations (\ref{eq:ent-compose}) and
(\ref{eq:var-compose}) of Proposition \ref{prop:ent-inf-var-compose}
completes the proof:
\begin{eqnarray*}
\bH^\mu[f^{\geq 1}] &=& \bH^\eta[F^{\geq 1}] +\sum_{S\neq\emptyset}
  \wt{F}(S)^2 \sum_{i\in S}\frac{\bH^\mu[g_i^{\geq
      1}]}{\Varx_\mu[g_i]} \\
&\leq& C\cdot (\Inf^\eta[F] - \Var_\eta[F]) + \sum_{S\neq\emptyset}
  \wt{F}(S)^2 \sum_{i\in S}\frac{\bH^\mu[g_i^{\geq
      1}]}{\Varx_\mu[g_i]} \\
&\leq & C\cdot (\Inf^\mu[f] - \Var_\eta[F]) \ = \ C\cdot (\Inf^\mu[f] -
\Var_\mu[f]).
\end{eqnarray*}
Here the first equality is by (\ref{eq:ent-compose}), the first
inequality by our second assumption, the second inequality by
(\ref{eq:inf-induct}), and finally the last identity by
(\ref{eq:var-compose}).
\end{proof}

\section{Distribution-independent bound for FEI$^+$}
\label{sec:product-FEI-bound}

In this section we prove that $\mu$-biased FEI$^+$ holds for all
Boolean functions $F:\bits^k_\mu\to\bits$ with factor $C$ independent
of the biases $\mu_1,\ldots,\mu_k$ of $\mu$. When $\mu = \la 0,\ldots
0\ra$ is the uniform distribution it is well-known that the FEI
conjecture holds with factor $C = O(\log k)$, and a bound of $C \leq
2^k$ is trivial since $\Inf[F]$ is always an integer multiple of
$2^{-k}$ and $\bH[F] \leq 1$; neither proofs carry through to the
setting of product distributions. We remark that even verifying the
seemingly simple claim ``there exists a universal constant $C$ such
that $\bH^\mu[\MAJ_3] \leq C\cdot (\Inf^\mu[\MAJ_3]-\Var_\mu[\MAJ_3])$
for all product distributions $\mu \in [0,1]^3$'', where $\MAJ_3$ the
majority function over 3 variables, turns out to be technically
cumbersome.  \medskip


The high-level strategy is to bound each of the $2^k-1$ terms of
$\bH^\mu[F^{\geq 1}]$ separately; due to space considerations we defer
the proof the main lemma to Appendix \ref{ap:omitted}.

\begin{lemma}
\label{lem:one-term-in-ent}
Let $F:\bits^k_\mu\to\bits$.  Let $S\sse [k]$, $S\neq\emptyset$, and
suppose $\wt{F}(S)\neq 0$.
For any $j\in S$ we have
\[ \wt{F}(S)^2 \log\left(\frac{\prod_{i\in
      S}\sigma_i^2}{\wt{F}(S)^2}\right) \leq \frac{2^{2k}}{
\ln 2} \cdot \Varx_\mu[D_{\phi^\mu_j}F]. \]
\end{lemma}

\begin{reptheorem}{thm:product-FEI-bound}
Let $F:\bits^k_\mu\to\bits$. Then
\[ \bH^\mu[F^{\geq 1}] \leq 2^{O(k)} \cdot
(\Inf^\mu[F]-\Var_\mu[F]). \]
\end{reptheorem}

\begin{proof}
The claim can be equivalently stated as $\bH^\mu[F^{\geq 1}] \leq
2^{O(k)} \sumi
\Var_\mu[D_{\phi^\mu_i}F]$, since
\[ \sumi\Var[D_{\phi^\mu_i}F] = \sum_{|S| \geq 2}|S|\cdot \wt{F}(S)^2
\leq 2\sum_{|S|\geq 2} (|S|-1)\cdot \wt{F}(S)^2 = 2\cdot (\Inf^\mu[F]
-\Var_\mu[F]).\] By Lemma \ref{lem:one-term-in-ent}, for every $S\neq
\emptyset$ that contributes $\varphi(\wt{F}(S))$ to $\bH^\mu[F^{\geq
  1}]$ we have $\varphi(\wt{F}(S)) \leq 2^{O(k)} \Var_\mu[D_{\phi_j^\mu}F]$,
where $j$ is any element of $S$.  Summing over all $2^k-1$ non-empty
subsets $S$ of $[k]$ completes the proof.
\end{proof}

\subsection{FEI$^+$ for read-once formulas}
\label{sec:fei-read-once}

Finally, we combine our two main results so far, the composition
theorem (Theorem \ref{thm:compose}) and the distribution-independent
universal bound (Theorem \ref{thm:product-FEI-bound}), to prove
Conjecture \ref{conj:prod-FEI++} for read-once formulas with arbitrary
gates of bounded arity.

\begin{definition}
  Let $\calB$ be a set of Boolean functions. We say that a Boolean
  function $f$ is a formula over the basis $\calB$ if $f$ is
  computable a formula with gates belonging to $\calB$.  We say that
  $f$ is a read-once formula over $\calB$ if every variable appears at
  most once in the formula for $f$.
\end{definition}

\begin{corollary}\label{cor:read-once-induction}
  Let $C > 0$ and $\calB$ be a set of Boolean functions, and suppose
  $\bH^\mu[F] \leq C\cdot (\Inf^\mu[F]-\Var_\mu[F])$ for all $F\in\calB$ and
  product distributions $\mu$. Let $\calC$ be the class of read-once
  formulas over the basis $\calB$.  Then $\bH^\mu[f] \leq C\cdot
  (\Inf^\mu[f]-\Var_\mu[f])$ for all $f\in\calC$ and product distributions $\mu$.
\end{corollary}

\begin{proof}
  We proceed by structural induction on the formula computing $f$. The
  base case holds since the $\mu$-biased Fourier expansion of the
  dictator $x_1$ and anti-dictator $-x_i$ is $\pm (\mu_1 +
  \sigma_1\phi_1^\mu(x))$ and so $\bH^\mu[f^{\geq 1}] =
  \wt{f}(\{1\})^2\log(\sigma_1^2/\wt{f}(\{1\})^2) = \sigma_1^2
  \log(\sigma_1^2/\sigma_1^2) = 0$.\medskip

  For the inductive step, suppose $f = F(g_1,\ldots,g_k)$, where
  $F\in\calB$ and $g_1,\ldots,g_k$ are read-once formulas over $\calB$
  over disjoint sets of variables. Let $\mu$ be any product distribution
  over the domain of $f$.  By our induction hypothesis we have
  $\bH^\mu[g_i^{\geq 1}] \leq C\cdot (\Inf^\mu[g_i]-\Var_\mu[g_i])$ for all
  $i\in [k]$, satisfying the first requirement of Theorem
  \ref{thm:compose}. Next, by our assumption on $F\in\calB$,
  we have $\bH^\eta[F^{\geq 1}] \leq C\cdot (\Inf^\eta[F]-\Var_\eta[F])$ for
  all product distributions $\eta$, and in particular, $\eta =
    \la\Ex_\mu[g_1],\ldots,\Ex_\mu[g_k]\ra$,
  satisfying the second requirement of Theorem
  \ref{thm:compose}.  Therefore, by Theorem
\ref{thm:compose} we conclude that $\bH^\mu[f]\leq C\cdot
(\Inf^\mu[f]-\Var_\mu[f])$.
\end{proof}

By Theorem \ref{thm:product-FEI-bound}, for any set $\calB$ of Boolean
functions with maximum arity $k$ and product distribution $\mu$, every
$F\in\calB$ satisfies $\bH^\mu[F]\leq 2^{O(k)}\cdot
(\Inf^\mu[F]-\Var_\mu[q])$.  Combining this with Corollary
\ref{cor:read-once-induction} yields the following:

\begin{reptheorem}{thm:read-once}
Let $\calB$ be a set of Boolean functions with maximum arity $k$, and
$\calC$ be the class of read-once formulas over the basis
$\calB$. Then $\bH^\mu[f] \leq 2^{O(k)}\cdot (\Inf^\mu[f]-\Var_\mu[f])$ for
all $f\in\calC$ and product distributions $\mu$.
\end{reptheorem}

\section{Lower bound on the constant of the FEI conjecture}
\label{sec:lb-unif-FEI}

The tools we develop in this paper also yield an explicit function $f$
achieving the best-known ratio between $\bH[f]$ and $\Inf[f]$ (\ie a
lower bound on the constant $C$ in the FEI conjecture).  We will use
the following special case of Proposition
\ref{prop:ent-inf-var-compose} on the behavior of spectral entropy and
total influence under composition:


\begin{lemma}[Amplification lemma]
 \label{lem:amplify}
Let $F:\bits^k\to\bits$ and $g:\bits^\ell\to\bits$ be balanced Boolean
functions. Let $f_0 = g$, and for all $m\geq 1$, define
$f_m = F(f_{m-1}(x^1),\ldots,f_{m-1}(x^k))$. Then
\begin{eqnarray*}
\bH[f_m] &=& \bH[g]\cdot
\Inf[F]^m + \bH[F] \cdot \frac{\Inf[F]^m-1}{\Inf[F]-1}  \\
\Inf[f_m] &=& \Inf[g]\cdot \Inf[F]^m.
\end{eqnarray*}
In particular, if $F= g$ we have
\[ \frac{\bH[f_m]}{\Inf[f_m]} = \frac{\bH[F]}{\Inf[F]} +
\frac{\bH[F]}{\Inf[F](\Inf[F]-1)} -
\frac{\bH[F]}{\Inf[F]^{m+1}(\Inf[F]-1)}. \]
\end{lemma}

\begin{proof}
  Since the composition of a balanced function with another remains
  balanced, we have the recurrence relations $\bH[f_m] =
  \bH[f_{m-1}]\cdot \Inf[F] + \bH[F] $ and $\bH[f_m] =
  \bH[f_{m-1}]\cdot \Inf[F] + \bH[F] $ as special cases of
  Proposition \ref{prop:ent-inf-var-compose}.
Solving them yields the claim.
\end{proof}

\begin{reptheorem}{thm:lb-unif-FEI}
There exists an infinite family of functions $f_m: \bits^{6^m}\to\bits$
such that
$\lim_{m\to\infty} \bH[f_m]/\Inf[f_m] \geq 6.278944$.
\end{reptheorem}

\begin{proof}
Let \[ g  =
( \overline{x}_1 \wedge x_2 \wedge x_3) \vee
( x_1 \wedge \overline{x}_2 \wedge x_4) \vee
( x_1 \wedge \overline{x}_2 \wedge x_5 \wedge x_6) \vee
( x_1 \wedge x_2 \wedge x_3) \vee
( x_1 \wedge x_2\wedge x_4\wedge x_5).
\]
It can be checked that $g$ is a balanced function with $\bH[F] \geq
3.92434$ and $\Inf[F] = 1.625$.  Applying Lemma \ref{lem:amplify} with
$F = g$, we get
\[  \lim_{m\to\infty}\frac{\bH[f_m]}{\Inf[f_m]} \ge
\frac{3.92434}{1.625} + \frac{3.92434}{1.625\times 0.625} =
6.278944.\]
\end{proof}

\bibliographystyle{alpha}
\bibliography{odonnell-bib}

\appendix

\section{Biased Fourier Analysis}
\label{ap:biased-fourier}

\begin{theorem}[Fourier expansion]
\label{thm:fourier-expansion}
Let $\mu=\la\mu_1,\ldots,\mu_n\ra$ be a sequence of biases.  The
$\mu$-biased Fourier expansion of $f:\bn\to\R$ is
\[ f(x) = \sumS \wtf(S)\phi^\mu_S(x), \]
where
\[ \phi^\mu_S(x) = \prod_{i\in S} \frac{x_i-\mu_i}{\sigma_i} \quad
\text{and} \quad \wtf(S) = \Ex_\mu[f(\bx)\phi^\mu_S(\bx)],\]
and $\sigma_i^2 = \Var_\mu[x_i] = 1-\mu_i^2$.
\end{theorem}

The
$\mu$-biased spectral support of $f$ is the collection $\calS\sse
2^{[n]}$ of subsets $S\sse [n]$ such that $\wtf(S) \neq 0$.  We write
$f^{\geq k}$ to denote $\sum_{|S|\geq k}
\wt{f}(S)\phi^\mu_S(x)$, the projection of $f$ onto its monomials of
degree at least $k$.

\begin{theorem}[Parseval's identity]
  Let $f:\bn_\mu\to\R$. Then $\sumS\wt{f}(S)^2 = \Ex_\mu[f(\bx)^2]$. In
  particular, if the range of $f$ is $\bits$ then $\sumS\wt{f}(S)^2 =
  1$.
\end{theorem}

\begin{definition}[Influence]
  Let $f:\bn_\mu\to\R$. The influence of variable $i\in [n]$ on $f$ is
  $\Inf_i^\mu[f] = \Ex_\rho[\Var_{\mu_i}[f_\rho]]$, where $\rho$ is a
  $\mu$-biased random restriction to the coordinates in $[n]\backslash
  \{i\}$.  The total influence of $f$, denoted $\Inf^\mu[f]$, is
  $\sumi\Inf_i^\mu[f]$.
\end{definition}

We recall a few basic Fourier formulas.  The expectation of $f$ is
given by $\Ex_\mu[f] = \wt{f}(\emptyset)$ and its variance $
\Var_\mu[f] = \sum_{S\neq\emptyset}\wtf(S)^2$. For each $i\in [n]$,
$\Inf^\mu_i[f] = \sum_{S\ni i} \wtf(S)^2$ and so $\Inf^\mu[f] = \sumS
|S|\cdot \wtf(S)^2$.  We omit the sub- and superscripts when $\mu =
\la 0,\ldots,0\ra$ is the uniform distribution. Comparing the Fourier
formulas for variance and total influence yields the Poincar\'e
inequality for functions $f:\bn_\mu\to\R$:

\begin{theorem}[Poincar\'e inequality]
\label{thm:poincare}
Let $f:\bn_\mu\to\R$. Then $\Inf^\mu[f] \leq \Var_\mu[f]$.
\end{theorem}

Recall that the $i$-th discrete derivative operator for $f\isafunc$ is
defined to be
\[ D_{x_i}(x) = {\lfrac 1 2}\left(f(x^{i\leftarrow 1}) - f(x^{i\leftarrow
    -1})\right),\]
and for $S\sse [n]$ we write $D_{x^S}f$ to denote $\circ_{i\in S}
D_{x_i}f$.
\begin{definition}[Discrete derivative]
The $i$-th discrete derivative operator $D_{\phi^\mu_i}$ with respect
to the $\mu$-biased product distribution on $\bn$ is defined by
 $D_{\phi^\mu_i}f(x) = \sigma_iD_{x_i}f(x)$.
\end {definition}
With respect to the $\mu$-biased Fourier expansion of $f:\bn_\mu\to\R$ the
operator $D_{\phi^\mu_i}$ satisfies
\[ D_{\phi^\mu_i}f = \sum_{S\ni i}\wtf(S) \phi^\mu_S, \]
and so for any $S\sse [n]$ we have   $\wtf(S) =
  \E[\circ_{i\in S}D_{\phi^\mu_i}f] = \prod_{i\in
    S}\sigma_i\Ex_\mu[(D_{x^S}f)]$.

\section{Omitted Proofs}
 \label{ap:omitted}

\begin{replemma}{lem:ent-inf-tensor}
Let $\Phi_1,\ldots,\Phi_k:\bits^{k\ell}_\mu\to\R$ where each $\Phi_i$
depends only on the $\ell$ coordinates in
$\{(i-1)\ell+1,\ldots,i\ell\}$.  Then
\[
 \bH^\mu[\Phi_1\cdots \Phi_k] = \sum_{i=1}^k
  \bH^\mu[\Phi_i]\prod_{j\neq i} \Ex_\mu[\Phi_j^2]
 \ \text{and} \ \
 \Inf^\mu[\Phi_1\cdots \Phi_k] =
 \sum_{i=1}^k \Inf^\mu[\Phi_i]
  \prod_{j\neq i} \Ex_\mu[\Phi_j^2]. \]
\end{replemma}

\begin{proof}
  We prove both formulas by induction on $k$, noting that the bases
  cases are trivially true. For the inductive step, we define $h(x) =
  \prod_{i\in [k-1]} \Phi_i(x)$ and see that
\begin{eqnarray*}
 \bH^\mu[h\cdot \Phi_k] &=&\mathop{\sum_{S\sse [(k-1)\ell]}}_{T\sse
   \{(k-1)\ell+1,\ldots
   k\ell\}} \wt{h}(S)^2 \wt{\Phi_k}(T)^2
 \log\left(\frac{\prod_{i\in S\cup
       T}\sigma_i^2}{\wt{h}(S)^2\wt{\Phi_k}(T)^2}\right) \\
&=&  \sum_{S,T} \wt{h}(S)^2 \wt{\Phi_k}(T)^2\left[
  \log\left(\frac{\prod_{i\in S}\sigma_i^2}{\wt{h}(S)^2}\right) +
\log \left(\frac{\prod_{i\in T}\sigma_i^2}{\wt{\Phi_k}(T)^2}\right) \right] \\
&=& \Ex_\mu[h^2]\cdot  \bH^\mu[\Phi_k] + \Ex_\mu[\Phi_k^2] \cdot \bH^\mu[h]  \\
&=& \prod_{i\in [k-1]} \Ex_\mu[\Phi_i^2] \cdot \bH^\mu[\Phi_k] + \Ex_\mu[\Phi_k^2]
\left(\sum_{i=1}^{k-1}
  \bH^\mu[\Phi_i]\prod_{j\neq i} \Ex_\mu[\Phi_j^2]  \right)\\
&=& \sum_{i=1}^k
  \bH^\mu[\Phi_i]\prod_{j\neq i} \Ex_\mu[\Phi_j^2].
\end{eqnarray*}
Here in the first equality we use the fact that if $f:\bn_\mu\to\R$ does
not depend on coordinate $i\in [n]$ then $\wtf(S) = 0$ for all $S\ni
i$ (\ie the Fourier spectrum of $f$ is supported on sets containing
only its relevant variables).  The third equality is by Parseval's,
and the fourth by the induction hypothesis applied to $h$. \medskip

The formula for influence follows from a similar derivation:
\begin{eqnarray*}
 \Inf^\mu[h\cdot \Phi_k] &=&\mathop{\sum_{S\sse [(k-1)\ell]}}_{T\sse
   \{(k-1)\ell+1,\ldots
   k\ell\}} |S\cup T| \cdot \wt{h}(S)^2 \wt{\Phi_k}(T)^2
 \\ &=&
\sum_{S,T} |T|\cdot
\wt{h}(S)^2 \wt{\Phi_k}(T)^2 + \sum_{S,T} |S|\cdot \wt{h}(S)^2
\wt{\Phi_k}(T)^2  \\
&=& \Ex_\mu[h^2]\cdot  \Inf^\mu[\Phi_k] + \Ex_\mu[\Phi_k^2] \cdot \Inf^\mu[h]  \\
&=& \prod_{i\in [k-1]} \Ex_\mu[\Phi_i^2]\cdot  \Inf^\mu[\Phi_k] + \Ex_\mu[\Phi_k^2]
\left(\sum_{i=1}^{k-1}
  \Inf^\mu[\Phi_i]\prod_{j\neq i} \Ex_\mu[\Phi_j^2]  \right) \\
&=& \sum_{i=1}^k
  \Inf^\mu[\Phi_i]\prod_{j\neq i} \Ex_\mu[\Phi_j^2],
\end{eqnarray*}
and this completes the proof.
\end{proof}

\begin{replemma}{lem:one-term-in-ent}
Let $F:\bits^k_\mu\to\bits$.  Let $S\sse [k]$, $S\neq\emptyset$, and
suppose $\wt{F}(S)\neq 0$.
For any $j\in S$ we have
\[ \wt{F}(S)^2 \log\left(\frac{\prod_{i\in
      S}\sigma_i^2}{\wt{F}(S)^2}\right) \leq \frac{2^{2k}}{
\ln 2} \cdot \Varx_\mu[D_{\phi^\mu_j}F]. \]
\end{replemma}

\begin{proof}
Recall that $\wt{F}(S) = \Ex_\mu[\circ_{i\in S}D_{\phi^\mu_i} f] =
\prod_{i\in S}\sigma_i \Ex_\mu[D_{x^S} f]$, and so
\begin{eqnarray*}
\wt{F}(S)^2 \log\left(\frac{\prod_{i\in S}\sigma_i^2}{\wt{F}(S)^2}\right)
  &=& \prod_{i\in S}\sigma_i^2\cdot
  \Ex_\mu[D_{x^S}F]^2\log\left(\frac1{\E[D_{x^S}F]^2}\right) \nonumber \\
&\leq& \frac1{\ln 2}\prod_{i\in S}\sigma_i^2 \cdot\big|\Ex_\mu[D_{x^S}F]\big|
\nonumber \\
&\leq& \frac1{\ln 2}\prod_{i\in S}\sigma_i^2 \Prx_\mu[D_{x^S}F \neq
0].  \nonumber
\end{eqnarray*}
Here the
first inequality holds since $t^2\log(1/t^2)\leq t/\ln(2)$
for all $t\in \R^+$, and the second uses the
fact that $D_{x^S}F$ is bounded within
$[-1,1]$. Therefore it suffices to argue that
\begin{eqnarray*} \prod_{i\in S}\sigma_i^2 \Prx_\mu[D_{x^S}F \neq 0]
  &\leq& 2^{2k}\cdot \Varx_\mu[D_{\phi^\mu_j}F] \\
&=& 2^{2k}\sigma_j^2\cdot \Varx_\mu[D_jF] \\ &=&
2^{2k}\sigma_j^2 \Ex_{y\in \bits^{[n]\backslash
  S}}\left[\Ex_{z\in\bits^{S\backslash\{j\}}}
\left[((D_jF)|_y(z)-\mu)^2\right]\right], \end{eqnarray*}
where $\mu = \E[D_jF]$ and $(D_jF)|_y$ denotes the restriction of
$D_jF$ where the coordinates in $[n]\backslash S$ are set according to
$y$.  We first rewrite the desired inequality above as
\[ \left(2^{-2k} \prod_{i \in S\backslash\{j\}}\sigma_i^2 \right)
\Ex_{y\in\bits^{[n]\backslash S}}[\ind_{D_{x^S}F(y) \neq 0}] \] \[ \leq
\Ex_{y\in \bits^{[n]\backslash
  S}}\left[\Ex_{z\in\bits^{S\backslash\{j\}}}
\left[((D_jF)|_y(z)-\mu)^2\right]\right]
\]
and argue that this holds point-wise: for
every $y\in [n]\backslash S$ such that $D_{x^S}F(y) \neq 0$,
\[ \E \left[((D_jF)|_y(z)-\mu)^2\right] \geq 2^{-2k}\prod_{i\in
  S\backslash\{j\}}\sigma_i^2. \] To see this, fix
$y\in\bits^{[n]\backslash S}$ such that $(D_{x^S}F)(y)\neq 0$.
Viewing $(D_{x^S}F)$ as $(D_{x^{S\backslash\{j\}}} D_jF)$, it follows
that $(D_jF)|_y$ is non-constant.  Since $(D_jF)|_y$ takes values in
$\{-1,0,1\}$, there must exist some $z^*\in\bits^{S\backslash\{j\}}$
such that $|(D_jF)|_y(z^*)-\mu|\geq \frac1{2}$ and so indeed
\begin{eqnarray*}
  \E
\left[((D_jF)|_y(z)-\mu)^2\right] &\geq& \left(\frac1{2}\right)^2
\Pr[z = z^*] \\ &=&
\frac1{4} \prod_{i\in
  S\backslash\{j\}}\frac{1\pm\mu_i}{2} \ \geq \
\frac1{4}\prod_{i\in S\backslash\{j\}}\frac{\sigma_i^2}{4} \ \geq\
2^{-2k}\prod_{i\in S\backslash\{j\}}\sigma_i^2. \end{eqnarray*}
\end{proof}
\end{document}